\documentclass[runningheads]{llncs}

\usepackage{epsfig}
\usepackage{latexsym}
\usepackage{enumerate}
\usepackage{url}
\usepackage{float}
\usepackage{texnansi}
\usepackage{color}
\usepackage{tikz}
\usepackage[margin=10pt,font=small,labelfont=bf]{caption}
\usepackage[subrefformat=parens,labelformat=parens]{subcaption}
\usepackage{afterpage}
\usepackage{enumitem}
\setlist[itemize]{label=\textbullet}
\usepackage[boxed]{algorithm}
\usepackage{algpseudocode}
\usepackage[normalem]{ulem}
\usepackage{lmodern}
\usepackage{booktabs}
\usepackage{sectsty}
\usepackage{ifthen}
\usepackage{amsmath}
\usepackage{amssymb}
\usepackage{amsfonts}
\usepackage{thmtools}
\usepackage{thm-restate}
\usepackage{mathtools}
\usepackage{xspace}
\usepackage{titling}
\usepackage{xfrac}
\usepackage{multirow}
\usepackage{bigdelim}
\usepackage{bm}

\usepackage[textsize=tiny]{todonotes}
\usepackage{booktabs} % For better horizontal lines
\usepackage{array}    % For table wrapping
\usepackage{tabularx} % For adjustable-width columns
\usepackage{authblk}
\usepackage{ragged2e}

\newcolumntype{L}{>{\raggedright\arraybackslash}p{0.3\textwidth}} % defining a new column type for wrapping
\newcolumntype{P}[1]{>{\centering\arraybackslash}p{#1}} % Centered column with fixed width
\usetikzlibrary{arrows,patterns,plotmarks,pgfplots.groupplots}
\tikzstyle{every picture} += [>=stealth]
\tikzset{axis/.style={semithick, line join=miter}}
\allsectionsfont{\sffamily}
\makeatletter
\def\@seccntformat#1{\csname the#1\endcsname.\quad}
\makeatother
\floatname{algorithm}{\normalfont\sffamily\bfseries Algorithm}
\floatstyle{ruled}
\provideboolean{fastcompile}
\newcommand{\hidefastcompile}[1]{\ifthenelse{\boolean{fastcompile}}{}{#1}}
\usepackage{pgfplots}
\usepackage{pgfplotstable}
\usetikzlibrary{calc}
\usepackage{mathtools}
\definecolor{orange}{rgb}{0.85,0.33,0.13} % 217,85,33
\definecolor{green}{rgb}{0.13,0.85,0.33}
\definecolor{purple}{rgb}{0.33,0.13,0.85}
\definecolor{lime}{rgb}{0.65,0.85,0.13}
\definecolor{blue}{rgb}{0,0,0.85}
\pgfplotscreateplotcyclelist{tricolor}{%
  orange,every mark/.append style={fill=orange!80!black},mark=*\\%
  green,every mark/.append style={fill=green!80!black},mark=square*\\%
  purple,every mark/.append style={fill=purple!80!black},mark=otimes*\\%
  black,mark=star\\%
  orange,every mark/.append style={fill=orange!80!black},mark=diamond*\\%
  green,densely dashed,every mark/.append style={solid,fill=green!80!black},mark=*\\%
  purple,densely dashed,every mark/.append style={solid,fill=purple!80!black},mark=square*\\%
  black,densely dashed,every mark/.append style={solid,fill=gray},mark=otimes*\\%
  orange,densely dashed,mark=star,every mark/.append style=solid\\%
  green,densely dashed,every mark/.append style={solid,fill=green!80!black},mark=diamond*\\%
}
\pgfplotsset{colormap={tricolormap}{color=(orange) color=(green) color=(purple)},
  colormap={quadcolormap}{color=(orange) color=(lime) color=(blue) color=(purple)}}
\pgfplotstableset{%
  font=\small,
  every head row/.style={before row=\toprule[1pt], after row=\midrule},
  every last row/.style={after row=\bottomrule[1pt]}}
\pgfplotsset{compat=1.15}
\usepackage{setspace}
\usepackage[margin=1in]{geometry}

\provideboolean{submissionversion}
\provideboolean{blindversion}
\setboolean{blindversion}{true}

\ifthenelse{\boolean{submissionversion}}{
  \renewcommand{\todo}[2][1]{}
  
  \newcommand{\deledit}[1]{}
}{
  
  \newcommand{\deledit}[1]{{\color{orange} \sout{#1}}}
}

\spnewtheorem{assumption}{Assumption}{\bfseries}{\itshape}

\usepackage{graphicx}
\usepackage{hyperref}
\usepackage{cleveref}
\usepackage{parskip}
\crefname{assumption}{assumption}{assumptions}

\title{\bf\sffamily Proof of Sampling: A Nash Equilibrium-Based Verification Protocol for
  Decentralized Systems}

  \author[1]{Yue Zhang\thanks{Equal contribution}}
    \author[1,2]{Shouqiao Wang$^\star$}
        \author[3]{Sijun Tan$^\star$}
    \author[3]{\authorcr Xiaoyuan Liu}
    \author[1,2]{Ciamac C. Moallemi}
    \author[1,3]{Raluca Ada Popa}
  \affil[1]{Hyperbolic Labs}
  \affil[2]{Columbia University}
  \affil[3]{UC Berkeley}
  \date{May 30, 2025}
\begin{document}
\maketitle

\begin{abstract}
  This paper introduces the Proof of Sampling (PoSP) protocol, a Nash Equilibrium-based verification mechanism, and its application to decentralized machine learning inference through spML.  
  Our protocol has a pure strategy Nash Equilibrium, compelling rational participants to act honestly. It economically disincentivizes dishonest behavior, making it costly for participants to compromise the network's integrity.
  
  In our spML protocol, we apply PoSP to decentralized inference for AI applications via a novel cryptographic protocol. The resulting protocol is much more efficient than zero knowledge proof based approaches. Moreover, we anticipate that the PoSP protocol could be effectively utilized for designing verification mechanisms within Actively Validated Services (AVS) in restaking solutions. 
  
  We further expect that the PoSP protocol could be applied to a variety of other decentralized applications. Our approach enhances the reliability and efficiency of decentralized systems, paving the way for a new generation of decentralized applications.

\end{abstract}
\section{Introduction}

In the development of decentralized protocols, it is customary to assume that honest nodes will
adhere to the established protocol. Take, for example, optimistic rollups
\cite{kalodner2018arbitrum} as a scaling solution for blockchain that aims to increase
blockchain’s transaction throughput. In this approach, a designated rollup validator processes
transactions off-chain and posts the results on the blockchain. By default, the transactions are
assumed to be executed correctly. However, if other validators detect incorrect transactions, they
can submit fraud proofs on-chain to challenge the validator who submitted the false results. The
security of optimistic rollup relies on the critical assumption that at least one rollup validator
is honest and validates the transactions. If all validators are dishonest, the on-chain
transactions could be fraudulent. To incentivize honest behavior, an economic reward system is
introduced, making it more profitable for rational nodes to act honestly. However, under this
assumption, Mamageishvil and Felten\cite{mamageishvili2023incentive} observed that in most existing decentralized
systems, the equilibrium state corresponds to a mixed strategy Nash Equilibrium, which means that
the optimal strategy for each validator is to cheat and not validate a certain percentage of the
time. %This implies a non-negligible probability of dishonest behavior that is never detected and
%introduces significant security vulnerabilities. \raluca{I commented this previous sentence based on Ciamac's point} \ccm{PoSP also has the property that a non-negligible percentage $1-p$ of inputs are never
 % validated, so it's difficult to critize other approaches for having ``non-negligible probability
  %of dishonest behavior that is never detected'', maybe make the comparison above sharper?}
%
In addition, \cite{mamageishvili2023incentive}
also proposed an alternate framework for optimistic rollups that mandates the scrutiny of the
system by all nodes. However, such a design introduces duplicated work from all nodes and
contravenes the fundamental principle of scalability, undermining the initial objective of Layer 2
design.

In this paper, we introduce the Proof of Sampling (PoSP) protocol designed to address these challenges effectively, under certain foundational assumptions. We expect that the PoSP protocol is applicable across a broad range of decentralized systems. In this paper, we introduce its application to decentralized AI inference platforms. Notably, our system achieves a pure strategy Nash Equilibrium, wherein each node outputs the correct result, a principle that is highlighted by \cite{kannan2024cryptoeconomics} as ideal for the design of secure decentralized systems. This signifies that when each node's strategy is directed towards maximizing its individual profit, the overall system maintains security.

\subsection{Main Contribution}

We made two contributions in this paper:

\begin{itemize}
    \item First, we propose Proof of Sampling (PoSP), a pure strategy Nash Equilibrium protocol to incentivize rational actors within the network to return the correct output, thereby strengthening the security and integrity of the decentralized protocol.
    \item Based on PoSP, we design spML, a concrete instantiation of the PoSP protocol for decentralized machine learning inference. Assuming all nodes act rationally, spML ensures that each node is incentivized to act correctly under some reasonable trust assumptions.
\end{itemize}

% \subsection{Design Principles}

% PoSP enhances decentralized system integrity by initially selecting an asserter who submits results without knowing which validators might review them. If challenged, $n$ validators are randomly selected, and only then do they learn each other's identities. This approach prevents initial collusion by ensuring the asserter does not know the identities of the validators, thereby eliminating the possibility of pre-arranged collusion with known associates.

% In the case where all nodes agree on the outcome, the system accepts this consensus, ensuring that the protocol functions correctly. However, if there are discrepancies, indicating dishonesty or error, those nodes are penalized. This mechanism ensures that acting honestly is not only the most straightforward but also the most beneficial strategy for all participants, thus maintaining the integrity and reliability of the system.

\subsection{Related Work}

\textbf{Game Theory for Protocol Analysis.} Employing game theory for protocol analysis, rational and economic theories have been extensively applied to explore various blockchain configurations, including Byzantine Fault Tolerance as discussed in \cite{amoussou2024committee,halaburda2021economic}, sharding strategies \cite{manshaei2018game}, proof-of-work systems \cite{badertscher2018but,garay2013rational,kiayias2016blockchain,luu2015demystifying}, proof-of-stake systems \cite{brunjes2020reward,saleh2021blockchain,schwarz2022three}, secure outsourced computation \cite{teutsch2024scalable}, and Layer 2 solutions \cite{jiang2025incentive,landis2023incentive,lee2025hollow,li2023security,mamageishvili2023incentive,sheng2024proof}. The use of slashing mechanisms, analyzed in \cite{kannan2024cryptoeconomics}, enhances network security and ideally aims to align node incentives in such a way that following the protocol-prescribed strategy always yields the most benefit, reflecting a unique Nash Equilibrium in pure strategies. A comprehensive review \cite{thibault2022blockchain} elaborates on the diverse applications of game theory across different blockchain scaling solutions. This analytical approach underscores the utility of game theory in dissecting and enhancing the security and scalability mechanisms of blockchain technology.

\textbf{Decentralized AI Inference.} The decentralized AI domain has significantly benefited from innovations like opML, as detailed by\cite{conway2024opml}, which enhances machine learning (ML) on the blockchain by facilitating the efficient processing of intricate models, alongside the introduction of Zero-Knowledge Machine Learning (zkML) by \cite{aziz2024zkvml,kang2022scaling,liu2021zkcnn,peng2025survey,sun2024zkllm,weng2021mystique,yadav2024fairproof,zhang2020zero}. These advancements in blockchain AI, opML and zkML, address scalability, security, and efficiency with distinct trade-offs. OpML enhances scalability and efficiency, yet its security may not be as robust. ZkML offers strong security through zero-knowledge proofs, yet faces challenges in scalability and efficiency due to high computational overhead.

\textbf{Game Design with Pure Strategy Nash Equilibrium.} Among the most relevant studies we reviewed, the following papers stand out for their design of game-theoretic models that achieve a unique Nash Equilibrium in pure strategies, where each node is incentivized to act honestly. This core design criterion is fundamental to our research as it aligns closely with our objectives of ensuring integrity and reliability in decentralized network environments. 

Outsourced computation within game theoretic frameworks, as discussed in studies \cite{belenkiy2008incentivizing,kupccu2015incentivized}, commonly involves a trusted third party to resolve disputes and deter collusion among coalitions. In our paper, we can use the arbitration protocol to ensure integrity without relying on a trusted third party.

Nix and Kantarcioglu \cite{nix2012contractual} implements a system where the integrity of outsourced data is maintained through a verification game involving two non-colluding cloud providers. In this model, the data owner begins by sending a query to one of the cloud providers. With a predefined probability, this query is also sent to a second provider for computation. The results from both providers are then compared to detect any discrepancies, effectively using one provider to challenge the results of the other. This approach assumes that the nodes do not collude and that there is a negligible probability of their results matching if both are incorrect. In contrast, our framework does not rely on these assumptions. We accommodate the possibility of collusion and does not rely on the assumption that non-colluding parties will always produce differing results if both are dishonest. This makes our protocol more realistic and applicable to real-world scenarios.

Dong et al.\cite{dong2017betrayal} discusses smart counter-collusion contracts for verifiable cloud computing; it utilizes game theory to foster distrust among potential colluders, effectively discouraging collusion by incentivizing betrayal among partners, making it economically disadvantageous and risky. While their model relies on the rational behavior of physically isolated nodes and assumes the presence of a trusted third party for dispute resolution, our approach offers a broader application by accommodating the possibility of irrational collusion and the control of multiple nodes by a single entity. Instead of relying on a centralized trusted third party, our model distributes trust to ensure integrity and resolve disputes, better aligning with the decentralized nature of blockchain environments and enhancing system robustness against a wider range of threats. This makes our framework more adaptable to real-world decentralized settings.

Lu et al.~\cite{lu2018enabling} presents a mechanism for decentralized blockchain systems without a trusted third party. In their model, if the results from multiple nodes differ, all nodes are penalized to prevent collusion and promote individual integrity. This approach addresses the issue that, in large networks, one party might control a significant number of nodes $n-1$, yet $n$ nodes are required to participate in each round. In contrast, our approach is more refined; it considers scenarios where a certain fraction of nodes might be controlled by one entity. In most cases, such as the application for decentralized AI inference network, we only employ a two-node system per round, randomly selecting a validator unknown to the asserter until after submission, which is a key point for our design. Additionally, our system uses blockchain-based arbitration to resolve discrepancies, ensuring that honest nodes are not unjustly penalized, thus maintaining fairness and enhancing reliability without relying on a trusted third party.

Pham et al.~\cite{pham2014optimal} outlines the use of a trusted third party for auditing outsourced computations, where discrepancies between nodes' results trigger an audit, and random checks ensure integrity even when results match. In contrast, our work  eliminates the need for a trusted auditor by leveraging distributed trust to arbitrate disputes only when results differ, enhancing decentralization and reducing computational overhead. We specifically address potential node collusion, assuming a scenario where up to a certain fraction of the network might collude or be controlled by a single party. Our paper not only simplifies the verification process by relying on result comparison and arbitration, which never happens if every node is rational, rather than correctness checks with a certain probability but also strategically mitigates collusion risks by maintaining validator anonymity until asserter's submissions are finalized.

The rest of this paper is organized as follows. In Section 2, we introduce PoSP protocol and prove that it has a unique Nash Equilibrium in pure strategies under certain conditions. In Section 3, we show a possible way to implement PoSP protocol in real world applications as an example. We design a sampling-based verification mechanism (spML), which is implemented by PoSP protocol, within a decentralized AI inference network, detailing the protocol design and its security validation. In Section 4, we conclude the paper and discuss the future extensions of our work.

%%% Local Variables:
%%% mode: latex
%%% TeX-master: "main"
%%% End:

\section{The PoSP Protocol}

The PoSP (Proof of Sampling) protocol is developed from a game theoretic perspective, focusing on the strategic interactions within a multi-agent system. This approach emphasizes the decision-making processes and the incentive structures that are essential for ensuring that the actions of all nodes, whether potentially Byzantine or inherently honest, align with the network’s objectives. Specifically, the PoSP protocol is structured to achieve a unique Nash Equilibrium in pure strategies, where every Byzantine participant is incentivized to act honestly, thereby enhancing the reliability and integrity of the system. While this section lays the theoretical groundwork, the specific application and implementation details are context-dependent. In the subsequent section, we will explore the application of PoSP in decentralized AI inference, specifically in a scenario termed spML, based on our PoSP protocol. Detailed implementation of spML will be discussed there.

\subsection{Protocol Design} \label{subsection:posp}
In this section, we propose the PoSP protocol. 
Consider a system with $N$ nodes. 
We assume that $f$ is a deterministic function, and $x$ is an input. 
An amount $B$ is paid for the following transaction computing $f$ on $x$. Let $R_A$ and $R_V$ be two positive numbers such that $R_A + R_V < B$.

\label{general-model}
\begin{enumerate}
    \item A node selected from the network serves as an asserter. This asserter calculates a value $f(x)$, with both the function $f(\cdot)$ and the input $x$ being well-known to the network, and outputs the result.

    \item With a predetermined probability $p$, a challenge protocol is triggered. If the challenge protocol is not triggered, this round concludes, and the asserter is awarded a reward denoted by $R_A$.
    
    \item If the challenge protocol is triggered, $n$ validators are randomly selected from the network of $N$ nodes, where $n\geq1$ is a predetermined integer parameter. Each validator, denoted as validator $i$, independently computes $f(x)$ and outputs the result.
    
    \item 
    If all results from the asserter and the validators match, the result is deemed valid and accepted. The asserter receives $R_A$ and each validator receives $R_V/n$. Otherwise, an arbitration protocol is initiated to determine the correctness of the asserter's result and that of each of the $n$ validators. We  assume that the arbitration process is always accurate. If the asserter's result is upheld by the arbitration protocol, it receives a reward of $R_A$; if not, it is penalized with $S$. Similarly, each validator deemed incorrect will also be penalized with $S$.
    Let $S_{\mathsf{total}}$ represent the total amount slashed from dishonest validators and the asserter. Each validator deemed honest, out of the $m$ deemed-honest ones from a total of $n$ validators, receives $R_V/m + S_{\mathsf{total}}/m$ if $m\geq 1$.
\end{enumerate}

In the PoSP protocol, the combined values of all rewards and penalties are carefully designed to ensure that the amount paid out never exceeds the amount collected. Excess tokens, arising from scenarios like untriggered challenges or accumulated penalties, are burned or go to the protocol reserve. This cautious approach also helps prevent manipulative behavior and maintains the system's integrity.

\subsection{Assumptions and Analysis}

Then we show that, under certain conditions, PoSP protocol has a unique Nash Equilibrium in pure strategies \ref{main-theorem} that all nodes output the correct result. This ensures there is no economic incentive for nodes to produce erroneous results. To analyze this protocol, we additionally define 

\begin{itemize}
    \item $C$:  the computational cost for a reference implementation of $f(x)$, established to ensure that all nodes allowed in the network can accurately compute the function within this cost
    \item $U_1$: maximum profit that the asserter can gain if he acts dishonestly and the challenge mechanism is not triggered
    \item $U_2$: maximum profit that the asserter can gain if the challenge mechanism is triggered and he controls all the validator nodes
\end{itemize}

\begin{assumption}
    \label{earn-more-than-cost}
    We assume that $S>nC$, $R_A-C>-S$ and $R_V/n-C>-S$.    
\end{assumption}
These inequalities imply that, it is more beneficial for each node to receive the reward than to be penalized. Then we show that given Assumption \ref{earn-more-than-cost}, the Byzantine validator who does not collude with the asserter has a dominant strategy to output the correct result.

\begin{table}[htbp]
\centering
\captionsetup{skip=10pt} 
\begin{tabular}{|P{3.8cm}|P{3.2cm}|P{3.6cm}|}
\hline
Strategy & Asserter Correct & Asserter Incorrect \\ \hline
Validator Correct & \( \geq R_V/n-C \) & \( \geq R_V/n -C + S/n \) \\ \hline
Validator Incorrect & \(-S\) & \(\leq R_V/n \) \\ \hline
\end{tabular}
\caption{Payoff for Validator in a Non-Collusive Scenario}
\label{not-collude-validator-payoff}
\end{table}

Table \ref{not-collude-validator-payoff} gives a game in a bimatrix format. The first row is the correctness of the asserter's outcome, and the first column is the correctness of the validator's outcome, who does not collude with the asserter. The value in the table is the payoff of the validator. Then, the following property is straightforward.

\begin{property}
\label{non-colluding-validator-always-honest}
    Under Assumption \ref{earn-more-than-cost}, if one validator does not collude with the asserter, its dominant strategy is to output the correct result.
\end{property}

\begin{assumption}
    \label{collusion-rate-assumption}
    We assume that at most a fraction $r$ of the total nodes in the network are Byzantine (that is, might deviate from the correct execution of $f$), which also implies that the fraction of the nodes in the network that the Byzantine asserter controls is at most $r$. 
\end{assumption}

\begin{theorem}
\label{main-theorem}
    Under Assumption \ref{earn-more-than-cost} and Assumption \ref{collusion-rate-assumption}, each asserter has a dominant strategy to act honestly and output the correct result if $$R_A+pS-(1-p)U_1-C>pr^n\left(U_2+S\right).$$
    This also means our system has a unique Nash Equilibrium in pure strategies, where each node behaves honestly and reliably outputs the correct result.
\end{theorem}

\begin{proof}
    The proof is provided in the Appendix \ref{appendix-proof}.
\end{proof}

%%% Local Variables:
%%% mode: latex
%%% TeX-master: "main"
%%% End:

\section{Application to Decentralized AI Inference Network Design}

Decentralized AI inference network has gained popularity in recent times due to rapid growth of AI and strong demand from the industry. Such a network leverages the computational power of a wide array of individual providers who contribute to the AI server pool in a permissionless manner. A well-designed decentralized AI inference network is able to balance the supply with the demand for AI inference capabilities. 

In a decentralized setting, AI inference services are performed by these distributed nodes. However, these nodes are not guaranteed to behave honestly, so simply executing a model and trusting its output is insufficient. For instance, if a user wants to perform AI inference using a powerful model like LLaMA3-70B, an execution node might choose to use a less capable model like LLaMA3-7B to save computational power. Therefore, it is crucial to design additional mechanisms that incentivize honest execution and penalize nodes that act dishonestly.

SpML is designed to verify the integrity of the system with a minimal increase in computational overhead for security purposes. The following sections will detail the application of the PoSP protocol in establishing a robust decentralized AI inference network.

\subsection{Design Principles \& Assumptions}

% \raluca{go through each assumption from posp and explain why it holds in this inference setting either because of the ML setting or because of our crypotgrapahic or security guarantees, e.g. BFT; arbitration is always accurate, f(x) is deterministic, assumption 1 and 2. Maybe after showing the protocol}

\noindent\textbf{\sffamily Deterministic ML Execution.} One critical assumption that PoSP makes, which spML inherits, is that the function $f$ must be a determinstic function. Achieving deterministic execution in ML is notoriously diffcult due to the inherent inconsistent nature of floating point arithmetics, as outlined in \cite{zheng2021agatha}. For example, multiplying two floating point numbers on different software/hardware configuration might render different reuslts.

To combat the inherent inconsistencies caused by floating-point calculations in ML, we follow similar practices from prior work \cite{zheng2021agatha,conway2024opml}, which implements fixed-point arithmetic and software-based floating-point libraries to ensure determinsitic ML execution. To fix the randomness, both the asserter and the validator will be assigned the same random seed. This approach ensures uniform, deterministic ML executions, enabling the use of a deterministic state transition function for the ML process, enhancing reliability in decentralized environments.

\medskip
\noindent\textbf{\sffamily Minimizing On-chain Operations.} Given the extensive usage in AI inference networks, processing or recording every AI inference outcome on the blockchain is impractical due to scalability constraints. Instead, AI inferences are computed off-chain by decentralized servers, which then relay the results and their digital signatures directly to users, bypassing the on-chain mechanism. On-chain operations are only necessary during arbitration, which can be avoided if all nodes act economically rational. This approach significantly reduces the blockchain's load while ensuring users receive authenticated and accurate inference results.

Critical functions, such as posting overall balance calculations at set intervals, are conducted on-chain to ensure transparency and security. Additionally, the protocol allows for on-chain handling of challenge mechanisms, enabling transparent and secure resolution of disputes or anomalies detected off-chain within the blockchain framework. This ensures both integrity and accountability in the system's operations.

% \medskip
% \noindent\textbf{\sffamily Stateless Design in ML Inference.} While applications such as inference might appear stateful to users—due to the ability to engage in ongoing interactions—they are fundamentally stateless. In our framework, each query is treated as independent; any necessary historical context is encapsulated within each new request, thereby maintaining statelessness throughout the ML process.

% \medskip
% \noindent\textbf{\sffamily Permissionless Network Participation.} We assume that the ML model abstracted as $f(x)$ is public knowledge, and any compute providers can join the network to host the model. 

% In our network, anyone can join the network, gain full access to the ML model \( f(x) \), and contribute by running an AI server. The function \( f(x) \), pivotal for the AI inference process, is established as common knowledge within the network. This inclusivity ensures that the model is capable for validation, promoting the security of the network.

\subsection{System Setup and Threat Model} 
Our system consists of users, executors, orchestrators and a blockchain. An executor chosen to run the query of a user is called an asserter and the one chosen to rerun the query of a user is called a validator.

We assume there are a total of $N$ executor nodes. Each executor is required to deposit at least $S$ in order to be considered as a valid executor. Out of these, we assume that at most a fraction $r$ of them are malicious. A malicious executor can behave in arbitrary ways including colluding with other malicious executors. 

We assume there are a total of $3f+1$ orchestrators and at most $f$ of them are malicious. The rest of them are trusted for availability and security. The orchestrators implement a Byzantine Fault Tolerance (BFT) consensus with state machine replication to tolerate the Byzantine behavior of the malicious nodes. In our protocol, we assume all actions taken from the orchestrator nodes are under BFT consensus. 

In addition, these orchestrator nodes jointly run a distributed random beacon protocol~\cite{sokbeacon} that tolerates up to $f$ malicious nodes to generate a public random value at the beginning of each epoch. Each node can query from this randomness beacon at epoch $t$ and obtain $\tau_t$ as the random seed. We assume that the random value is a block cipher key (and if it is not, one can apply a key derivation function). 

We further assume that the blockchain is trusted for availability and security. It maintains a trusted list of all executors, a trusted PKI that stores public keys for all users and executors and the balances of all parties. At the beginning of the protocol, users and executor nodes query the trusted PKI to obtain copies of all orchestrator nodes' public keys. We assume that executors and users have the correct list of public keys for the orchestrators and we discussed in the next version of this document how to perform membership changes. The orchestrator nodes also keep copies of the public keys of all registered users and executor nodes. 

Communication between any two entities in this system happens over TLS.

\subsection{Notations and Cryptographic Building Blocks.}

\newcommand{\Seed}{\mathsf{Seed}}
\newcommand{\seed}{\mathsf{seed}}
\newcommand{\pk}{\mathsf{pk}}
\newcommand{\user}{\textsf{user}}
\newcommand{\Commit}{\textsf{Commit}}

% We use $\Commit(x, r)$ to denote a cryptographic commitment~\cite{blum1983coin,damgaard1998commitment} to input $x$ (with $r$ as fresh randomness). The cryptographic commitment is both hiding and binding.

We use $\textsf{PRF}$ to denote a pseudorandom function~\cite{prf}, which is a function that can be used to generate output from a random seed and a data variable, such that the output is computationally indistinguishable from truly random output. We define two deterministic sampling functions, $\texttt{Bucket} : \mathcal{S} \times \mathbb{N} \to {1, 2, \dots, N}$ and $\texttt{Random} : \mathcal{S} \to [0, 1)$, as follows:

% We use $\texttt{Hash}$ to denote a cryptographic hash function. Let $\Seed$ be a random variable representing the random seed uniformly distributed over a suitable range $\mathcal{S}$. We define two deterministic sampling functions, $\texttt{Bucket} : \mathcal{S} \times \mathbb{N} \to {1, 2, \dots, N}$ and $\texttt{Random} : \mathcal{S} \to [0, 1)$, as follows:
\begin{itemize}
    \item 
$\texttt{Bucket}(\seed, \mathsf{string}, N) = \mathsf{PRF}_{\seed}(\mathsf{string}) \mod N$: maps a random block cipher key $\mathsf{seed}$ and a unique string to a random but deterministic value in $\{1, \dots, N\}$
    
    \item $\texttt{Sampled}(\seed, \mathsf{uniquestr}, p) = [ \mathsf{PRF}_{\seed}(\mathsf{uniquestr}) < p \times \mathsf{PRFMax}$]: returns true with approximately probability $p$ based on a random block cipher key $\mathsf{seed}$ and a unique string.
\end{itemize}

To denote a digital signature over the message $x$, we use $\sigma_x$. To prevent replay attacks, each user request is assigned a unique request ID (\textsf{reqid}). We assume that every digital signature \(\sigma_x\) is generated over the combination of the original message \( x \) and this \textsf{reqid}.

\subsection{The SpML protocol} 

In this section, we describe spML, a concrete instantiation of the PoSP protocol to enable verifiable machine learning inference. 

% Each transaction includes the user's address, the asserter's address, the amount of transaction fee, the validator's address if challenged and the signatures by orchestrator nodes \yz{TODO}. 
The orchestrators have a record of all status of pending transactions (the transactions that haven't been finalized and posted to the blockchain yet) and the updates of the balances of the users and the executors. %The orchestrator will post and execute transactions in a FIFO manner.

The user has an input $x$ and wants to compute $f(x)$, where $f$ is a machine learning model. The spML protocol \label{ai-protocol}is designed as follows.

\paragraph{Basic Protocol}
\begin{description}
\item[Step 1: User Request Submission]
The user sends their request $(x, \textsf{user\_nonce})$ and its signature $\sigma_{(x, \textsf{user\_nonce})}^{\textsf{user}}$ to all orchestrator nodes.
    \begin{enumerate}
        \item Each orchestrator node queries $\pk_{\textsf{user}}$ from the PKI and calculates $\textsf{reqid} := \textsf{PRF}(\pk_{\textsf{user}} || x || \textsf{user\_nonce})$.
        \item Verifies that $\sigma_{(x, \textsf{user\_nonce})}^{\textsf{user}}$ is a valid signature of $(x, \textsf{user\_nonce})$.
        \item Orchestrators engage in a BFT agreement to accept this \textsf{reqid}. This agreement ensures that the \textsf{reqid} has not been processed before and prevents duplicate processing. If the BFT agreement is successful and the signature is valid, the request $(x, \textsf{reqid})$ is accepted and stored in the pending transaction pool of each orchestrator node.
    \end{enumerate}

\item[Step 2: Orchestrator Agreement and Asserter Selection]
The orchestrators agree (through BFT) on processing the user's request $(x, \textsf{reqid})$ at the current request processing epoch $t_{\textsf{req}}$. Each of them does the following:
\begin{enumerate}
    \item Queries the random seed $\tau_{t_{\textsf{req}}}$ based on this epoch and reuses the stored $\pk_{\textsf{user}}$ from Step 1.
    % \raluca{shouldn't they already have PK user from the request submission and added to the log in reqid?}
    \item Calculates $i := \texttt{Bucket}(\tau_{t_{\textsf{req}}}, \pk_{\textsf{user}}||x||\textsf{reqid},N)$ to select the asserter (node $i$).
    \item Signs over $(x, \textsf{reqid})$ (denote orchestrator $k$'s signature as $\sigma_{(x,\textsf{reqid})}^k$) and sends $(x, \textsf{reqid}, \sigma_{(x,\textsf{reqid})}^k)$ to the selected asserter, node $i$.
\end{enumerate}

\item[Step 3: Asserter Execution and Response]
The asserter, node $i$, upon receiving messages from orchestrators, does the following (within a defined timeout $T_{assert}$):
\begin{enumerate}
    \item Verifies $\sigma_{(x,\textsf{reqid})}^k$ is a valid signature over $(x, \textsf{reqid})$ signed by the orchestrator $k$.
    \item If the asserter receives at least $2f+1$ such verified messages for the same $(x, \textsf{reqid})$ from distinct orchestrators, it accepts the request and calculates $y_i = f(x)$.
    \item Signs the result: $\sigma_i := \textsf{Sign}_{\textsf{asserter}_i}(x, \textsf{reqid}, y_i)$.
    \item Sends $(x, \textsf{reqid}, y_i, \sigma_i)$ back to all orchestrator nodes.
    \item Handling Asserter Failure: If node $i$ fails to respond to orchestrators with a valid $(x, \textsf{reqid}, y_i, \sigma_i)$ within timeout $T_{assert}$, orchestrators will BFT-agree on this failure. The request will be re-assigned and the asserter may be penalized (see Timeout and Failure Handling Protocol).
\end{enumerate}
    Each orchestrator node, upon receiving $(x, \textsf{reqid}, y_i, \sigma_i)$ from asserter $i$:
    \begin{enumerate}
        \item Verifies the asserter's signature $\sigma_i$ on $(x, \textsf{reqid}, y_i)$.
        \item If valid, forwards $(y_i, \sigma_i, \textsf{reqid})$ to the user.
    \end{enumerate}
The user can proceed to consume the result $y_i$ once they receive $(y_i, \sigma_i, \textsf{reqid})$ from at least $2f+1$ orchestrators, all containing the same $y_i$ value for the given $\textsf{reqid}$, and after personally verifying at least one valid $\sigma_i$.

\item[Step 4: Orchestrator Post-Execution and Challenge Decision]
The orchestrators verify the asserter's signature $\sigma_i$ (if not already done). They then agree (through BFT) on processing the asserter's result $(y_i, \sigma_i)$ for request $\textsf{reqid}$. This occurs before or at the start of the challenge decision epoch $t_{\textsf{chal}}$. Each orchestrator then does the following:
\begin{enumerate}
    \item Queries the random seed $\tau_{t_{\textsf{chal}}}$ based on this epoch $t_{\textsf{chal}}$ and $\pk_{\textsf{user}}$ from the PKI.
    \item Calculates $b := \texttt{Sampled}(\tau_{t_{\textsf{chal}}}, \pk_{\textsf{user}}||x||\textsf{reqid}, p)$, where $p$ is a pre-determined threshold.
    \item If $b = 1$, then it initiates the Challenge Protocol for $(x, \textsf{reqid}, y_i, \sigma_i)$. Otherwise, it records a reward of $R < B/2$ for the asserter $i$ locally (associated with $\textsf{reqid}$), and concludes this request's active processing.
\end{enumerate}
\end{description}

\paragraph{Challenge Protocol}
\begin{description}
\item[Step 1: Initiation and Validator Selection]
If the Challenge Protocol is initiated for $(x, \textsf{reqid}, y_i, \sigma_i)$, each orchestrator does the following:
\begin{enumerate}
    \item Uses the same random seed $\tau_{t_{\textsf{chal}}}$ (from Basic Protocol Step 4.1).
    \item Calculates $j := \texttt{Bucket}(\tau_{t_{\textsf{chal}}}, \pk_{\textsf{user}}||x||\textsf{reqid}, N)$ to select the validator (node $j$).
        (Note: Asserter $i$ cannot be selected as validator $j$ for the same request. If $j=i$, a deterministic rule, e.g., $j=(i+1) \pmod N$, should apply.)
    \item Sends $(x, \textsf{reqid})$ along with its orchestrator signature $\sigma_{(x,\textsf{reqid})}^k$ (from Basic Protocol Step 2.3) to the validator, node $j$.
\end{enumerate}

\item[Step 2: Validator Execution and Response]
The validator, node $j$, upon receiving messages from orchestrators, does the following (within a defined timeout $T_{validate}$):
\begin{enumerate}
    \item Verifies $\sigma_{(x,\textsf{reqid})}^k$ is a valid signature over $(x, \textsf{reqid})$ signed by the orchestrator $k$.
    \item If the validator receives at least $2f+1$ such verified messages for the same $(x, \textsf{reqid})$ from distinct orchestrators, it accepts the task.
    \item Calculates $y_j = f(x)$.
    \item Signs the result: $\sigma_j := \textsf{Sign}_{\textsf{validator}_j}(x, \textsf{reqid}, y_j)$.
    \item Sends $(x, \textsf{reqid}, y_j, \sigma_j)$ back to all orchestrator nodes.
    \item Handling Validator Failure: If node $j$ fails to respond to orchestrators with a valid $(x, \textsf{reqid}, y_j, \sigma_j)$ within timeout $T_{validate}$, orchestrators will BFT-agree on this failure. The request will be reassigned and Node $j$ may be penalized. (see Timeout and Failure Handling Protocol).
\end{enumerate}

\item[Step 3: Orchestrator Verification and Outcome]
Each orchestrator, upon receiving $(x, \textsf{reqid}, y_j, \sigma_j)$ from validator $j$:
\begin{enumerate}
    \item Verifies the validator's signature $\sigma_j$ on $(x, \textsf{reqid}, y_j)$.
    \item If valid, it compares $y_j$ with the previously received $y_i$ for the same $\textsf{reqid}$.
    \item If $y_i = y_j$: Records a reward of $R < B/2$ for each of the asserter $i$ and the validator $j$ locally (total $2R < B$), and concludes this request's active processing.
    \item Otherwise (if $y_i \neq y_j$): Initiates the Arbitration Protocol, providing $(x, \textsf{reqid}, y_i, \sigma_i, y_j, \sigma_j)$.
\end{enumerate}
\end{description}

\paragraph{Arbitration Protocol}
\begin{description}
\item[Step 1: Orchestrator Request to Smart Contract]
If the Arbitration Protocol is initiated, each of the orchestrators sends an arbitration request containing $(\texttt{Arbitration}, x, \textsf{reqid}, y_i, \sigma_i, y_j, \sigma_j)$ along with its own signature on this entire tuple, to the arbitration smart contract on the blockchain.

\item[Step 2: Smart Contract Verification and Judgment]
The smart contract does the following:
\begin{enumerate}
    \item Queries the public keys of the orchestrators and verifies the signatures on the arbitration requests.
    \item If at least $2f+1$ valid and identical arbitration requests are received from distinct orchestrators, it continues. Otherwise, it aborts or awaits more identical requests.
    \item Verifies that $\sigma_i$ is a valid signature of asserter $i$ on $(x, \textsf{reqid}, y_i)$.
    \item Verifies that $\sigma_j$ is a valid signature of validator $j$ on $(x, \textsf{reqid}, y_j)$.
    \item It computes (or securely verifies via ZKP) the true result $y_{true} = f(x)$.
    \item Identifies honest and dishonest parties:
        \begin{itemize}
            \item If $y_i = y_{true}$ and $y_j \neq y_{true}$: Asserter $i$ is honest, validator $j$ is dishonest.
            \item If $y_i \neq y_{true}$ and $y_j = y_{true}$: Asserter $i$ is dishonest, validator $j$ is honest.
            \item If $y_i \neq y_{true}$ and $y_j \neq y_{true}$ (and $y_i \neq y_j$): Both are dishonest. (If $y_i=y_j \neq y_{true}$, challenge should not have led to arbitration, but if it does, both made the same error).
        \end{itemize}
    \item Slashes the node(s) that generated the wrong result(s) with a certain amount $S$ of their deposit.
    \item Records the arbitration outcome (who was correct/incorrect, slashed amounts).
\end{enumerate}

\item[Step 3: Orchestrator Post-Arbitration Processing]
Each orchestrator queries the arbitration result from the smart contract for $\textsf{reqid}$.
\begin{itemize}
    \item If asserter $i$ was deemed honest, records a reward (e.g., $R_{base} + S_{collected\_from\_j}$) for asserter $i$.
    \item If validator $j$ was deemed honest, records a reward (e.g., $R_{base} + S_{collected\_from\_i}$) for validator $j$.
    \item If both were dishonest, their deposits $S$ are slashed.
\end{itemize}
The user is also informed of the arbitrated correct result $y_{true}$.
\end{description}

\paragraph{Timeout and Failure Handling Protocol}
\begin{description}
    \item[Case 1: Asserter Fails to Respond (Basic Protocol Step 3)]
    If asserter $i$ fails to provide a valid signed result within $T_{assert}$:
    \begin{enumerate}
        \item Orchestrators BFT-agree on the timeout.
        \item Reassign asserter: Orchestrators may select a new asserter $i'$ (e.g., $i' := \texttt{Bucket}(\tau_{t_{\textsf{req}}}, \pk_{\textsf{user}}||x||\textsf{reqid}|| \textsf{attempt\_2}, N)$). The protocol restarts for this request from Basic Protocol Step 2.3 with the new asserter. Node $i$ may be penalized (e.g., small penalty recorded locally, or via smart contract if persistent).
    \end{enumerate}
    \item[Case 2: Validator Fails to Respond (Challenge Protocol Step 2)]
    If validator $j$ fails to provide a valid signed result within $T_{validate}$:
    \begin{enumerate}
        \item Orchestrators BFT-agree on the timeout.
        \item Reassign Validator: Similar to asserter failure, the protocol reassigns another validator to complete this request, and validator $j$ may be penalized.
    \end{enumerate}
\end{description}

\paragraph{Settlement Protocol}
This protocol is executed periodically, e.g., at the end of each settlement epoch $t_{\textsf{settle}}$.
\begin{description}
\item[Step 1: Orchestrator Agreement on Balance Updates]
After each settlement epoch, each orchestrator:
\begin{enumerate}
    \item Identifies all transactions (identified by $\textsf{reqid}$) that have concluded (either via Basic, Challenge, or Arbitration outcomes) within this epoch and whose financial implications (rewards, slashes) have not yet been posted on the blockchain.
    \item Calculates the net updates to the balances of all involved users and nodes (asserters, validators) based on recorded rewards and slashes.
    \item Crucially, all orchestrators engage in a BFT agreement protocol to arrive at a single, consistent list of balance updates for the epoch.
\end{enumerate}

\item[Step 2: Submission to Smart Contract]
Once BFT agreement on the balance updates is reached:
\begin{enumerate}
    \item One or more (or all) orchestrators submit this agreed-upon batch of balance updates, along with evidence of the BFT agreement (e.g., $2f+1$ signatures from orchestrators on the hash of the batch), to the settlement smart contract on the blockchain.
    \item The smart contract verifies the BFT agreement proof (e.g., the $2f+1$ orchestrator signatures).
    \item If verified, the smart contract applies the batch of balance updates to the users' and nodes' accounts managed by it.
\end{enumerate}
\end{description}

\subsection{Discussion} \label{s:discussion}

% \raluca{Add a discussion on the different tiers within which you have the same compute power and the same cost. }

% \paragraph{Inconsistent Response from the Asserter and Orchestrator} To reduce query latency, we ask the asserter to send the result immediately back to the user, and send the commitment of the result to the orchestrators. The asserter $i$ could, in theory, send a wrong output $y'$ to the user but sends the correct commitment $y_i = f(x)$ to the orchestrator. Our protocol provides no punishment for this inconsistent behavior from the asserter. However, in the next epoch, the orchestrators will send the asserter's revealed results back to the user, and there are guarantees that if this result is not correct, the asserter will get punished.

\paragraph{Malicious Users} Malicious users could collude with malicious executors. If the selected asserter is malicious, then a malicious user could collude with the asserter so that the asserter gains rewards for computing nothing $1-p$ percent of the time. However, the user has to pay for the query and our protocol guarantees that the net reward for the malicious user and asserter is negative.

\paragraph{Malicious Executors} Assuming that a fraction of $r$ executor nodes are malicious, then the probability of selecting two malicious executors is $r^2$. In theory, our user could receive two incorrect but same results from the asserter and the validator if they are both malicious. 

However, when the asserter responds to the orchestrator, it has no knowledge of and cannot control which node will be selected next. Therefore, if the asserter wants to cheat and provide a wrong response, it can only pass the check $(1-p)+pr$ percent of the time (either the challenge does not happen, or the challenge protocol chooses another malicious node). 

Even though malicious executors could choose not to follow the protocol, our mechanism design guarantees that the net reward for them is negative, and any economically rational executor would behave honestly.

\paragraph{Malicious Orchestrators} In our protocol, up to $f$ orchestrator nodes could be malicious. These malicious orchestrators cannot interfere with the correctness of the BFT concensus since there are $2f+1$ honest nodes. They could, however, collude with malicious users and executors. If a malicious validator is selected, the malicious orchestrator node could tell the validator the computed f(x) from the asserter. Then, the malicious validator could free-ride the asserter's result and save the computational cost of computing $f(x)$. Our protocol does not prevent against this cost-saving behavior. However, this behavior does not impact the security of our protocol, since our mechanism design guarantees that any economically rational asserter will compute the correct result. The only case where the user fails to obtain the correct result is when both the asserter and validator are malicious, and the asserter computes the wrong result. As we analyze above, the expected reward for this asserter is negative in this case.  

% Since our mechanism design guarantees that 

% Our design does not prevent against this scenario, but the  

% This condition is taken into consideration in our Nash equilibrium calculation (see Proposition \ref{ai-prop}), where we assume that at most a fraction $r$ of the executor nodes are malicious and can collude with each other.

% % this aligns with Assumption \ref{collusion-rate-assumption} and hence Proposition \ref{ai-prop} in Section \ref{subsection:analysis} holds.

\paragraph{Unresponsive Executors} In our protocol, we assume all executor nodes are available up to a certain limit. An executor node can reject a request if the number of requests it is concurrently handling exceeds this limit, but it is expected to respond to user requests if the limit is not met. If an asserter or validator does not respond within a specified number of epochs, the orchestrators generate a new random seed from the random beacon and send the requests to a new asserter or validator, potentially penalizing the unresponsive node. This setup incentivizes executor nodes to respond promptly. Even if an executor does not respond, the only negative effect for the user is a certain period of delay.

\paragraph{Multi-tier Executors} For a network of nodes with heterogeneous compute power, they can be grouped into tiers of executors based on compute power. This ensures balanced workloads and fair resource allocation, with consistent performance within each tier. Users can choose different tiers based on their requirements and budget, optimizing for either high computational power or cost-efficiency. In the spML protocol, orchestrators assign tasks within the same tier to make sure both the asserter and the validator are able to compute the user's request $f(x)$, hence enhancing verification robustness and network integrity. 

\paragraph{Aggregating Signatures for Reduced Cost}The current protocol design requires an orchestrator to send multiple signatures to the blockchain. It can be reduced to only one signature by using a threshold signature scheme.

\subsection{Analysis} \label{subsection:analysis}

% \sq{TODO}

% \begin{property}
% \label{user-behavior}
% The rational user will never collude with the server in the network, and the dominant strategy for the user is always act honestly.
% \end{property}

% This is because it is always impossible to slash the honest node. Then we can assume that the user will not collude with any server in the network. Moreover, if the challenge mechanism is triggered, the rational user will indeed send his input to another server B, because he has the utility to gain the transaction fee discount and get to know the true outcome of his request.

% Given Property \ref{user-behavior}, we can assume that the user will always act honestly. In this case, the orchestrator, i.e. the user, is honest, and the spML protocol in Section \ref{ai-protocol} will reduce to the PoSP protocol in Section \ref{general-model}.

\begin{proposition}
\label{ai-prop}
If $$p>\frac{C}{(1-r)S+(1-2r)R},$$ the system has a unique Nash Equilibrium in pure strategies, where every participant outputs the correct result.
\end{proposition}

\begin{proof}

After plugging $n=1$, $R_A=R$, $U_1=R$ and $U_2=2R$ in Theorem \ref{main-theorem}, we can get this result.

\end{proof}

As you can see from Proposition \ref{ai-prop}, the numerator equals to the computational cost for running one ML model, which is considered to be much less than the denominator. This means if we design the value of the reward and penalty appropriately, we only need little extra computational overhead to guarantee the security of the network. 

\subsection{SpML vs. Existing Decentralized AI Solutions}
In this section, we compare spML with the two prevalent methodologies in decentralized AI networks: optimistic fraud proof based approach (opML) and zero knowledge proof based approach (zkML).

\medskip
\noindent\textbf{\sffamily OpML.} Contrary to the heavy cryptographic reliance of zkML, opML adopts a fundamentally different strategy based on dispute resolution mechanisms. The optimistic approach presupposes that participants will act honestly, given the economic disincentives for fraudulent behavior. In the rare event of disputes, opML provides mechanisms for challenge and resolving fraudulent claims, ideally without necessitating heavy computational verification for every transaction. Nevertheless, the reliance on economic incentives and dispute resolution may introduce vulnerabilities for network security.

\medskip
\noindent\textbf{\sffamily ZkML.} At its core, zkML leverages zero-knowledge proofs. In the context of decentralized AI, zkML ensures that computations can be verified for correctness without revealing the underlying data or the specifics of the computation. This characteristic is particularly advantageous for applications requiring stringent data privacy measures. However, the sophistication and computational intensity of generating zero-knowledge proofs present challenges in terms of efficiency and accessibility. 

% \raluca{we have some delays until the randomness beacon epoch, and then the fraud and zK proofs -- maybe you want to say that it is a known time not undetermined?}
% \sq{It is indeed a known time not undetermined, but what we want to mention is that the user do not need to wait for the result provided by the validator in the challenge mechanism and trust the result provided by the asserter, because the asserter is always honest if everyone is rational}

\begin{table}[htbp]
\captionsetup{skip=10pt} 
\centering
\label{ml-comparison}
\begin{tabular}{|>{\Centering}m{2.3cm}|>{\Centering}m{4.0cm}|>{\Centering}m{4.5cm}|>{\Centering}m{5.2cm}|}
\hline
Aspect & opML & zkML & spML \\ 
\hline
Security & Relies on AnyTrust assumption (no pure-strategy Nash equilibrium) & High security through cryptographic proofs &  Security through economic incentives (pure-strategy Nash equilibrium) \\ 
\hline
Delays & Potential delays in dispute resolution & Delays due to proof generation & Almost no delay since rational asserter will act honestly \\ 
\hline
Efficiency & Highly efficient & Limited by computational overhead of proof generation & Highly efficient  \\ 
\hline
% Simplicity & Simple, unless for fraud proof & Complex due to zk proofs & Simple \\ 
% \hline
Overhead & Low computational overhead, unless in the case of disputes & High computational overhead due to the nature of cryptographic proof generation & Low computational overhead, unless in the case of disputes which never happen if everyone is rational \\ 
\hline
\end{tabular}
\caption{Comparison of OpML, ZkML and SpML}
\label{comparison-verification-mechanism}
\end{table}

\noindent Table \ref{comparison-verification-mechanism} compares opML, zkML and spML.

\medskip
\noindent\textbf{\sffamily Security.} The security of opML relies on the AnyTrust assumption. OpML only has a mixed strategy Nash Equilibrium, which means that there is a positive probability for undetected fraud if every node is rational. Conversely, zkML boasts robust security due to its use of cryptographic proofs. The security of spML is based on economic incentives. In spML, the initiation of challenge mechanism is an automated process managed by the protocol itself, rather than relying on the assumption that there will be at least one external validator, as is the case with opML. The pure strategy Nash Equilibrium demonstrated by spML offers evidence that the system can be deemed secure, provided that each node behaves rationally.

\medskip
\noindent\textbf{\sffamily Delays.} In opML, delays exist due to the challenge period, during which a transaction can be challenged with a fraud proof. This is a drawback in scenarios requiring real-time results. zkML faces inherent significant delays due to the computational overhead in proof generation. SpML is designed to mitigate delay issues altogether. Even if the challenge mechanism is triggered, the user does not need to wait for the challenge procedure: the user can trust the result because the dominant strategy for the asserter is to output the correct result, if every node is rational.

\medskip
\noindent\textbf{\sffamily Efficiency.} OpML is recognized for its efficiency, especially when disputes are minimal, suggesting a lightweight protocol suitable for extensive applications. ZkML's efficiency is hampered by the heavy computational load required for proof generation. In contrast, spML is presented as highly efficient as well, which can handle extensive network activity without significant degradation in performance.

% \medskip
% \noindent\textbf{\sffamily Simplicity.} The complexity of ZK proofs contributes to the high complexity of zkML, possibly making it less accessible for broader implementation. For opML, it is usually considered to be simple. However, its implementation can become complex when applied to fraud proof scenarios. In contrast, spML maintains consistent simplicity in implementation due to its arbitration procedure based on voting. This simplicity enhances ease of integration and maintenance, potentially facilitating widespread adoption.

\medskip
\noindent\textbf{\sffamily Overhead.} OpML claims low computational overhead, with the caveat that opML may incur higher overhead during disputes. ZkML's approach results in high computational overhead due to cryptographic processes. SpML also has a low computational overhead. During the challenge mechanism, spML still has a low computational overhead. This is because in spML, the challenge mechanism happens very rarely. Only when the results do not match during the challenge mechanism, spML may incur high overhead during arbitration, but the arbitration never happens if every node is rational.

\subsubsection{Empirical Evaluation}
For this part, we use empirical evaluation to further compare opML, zkML and spML.

For zkML, existing solutions, as demonstrated in \cite{ganescu2024trust,kang2022scaling}, indicate that generating a proof for a nanoGPT model with $1$M parameters takes approximately $16$ minutes. However, for more advanced models like Llama2-70B, which possesses $70,000$ times more parameters than nanoGPT, it is reasonable to expect that generating a single proof could take several days or weeks. Consequently, employing zkML in a decentralized AI inference network may not be practical given the extended time requirements.

In the opML scenario, when the validator initiates the fraud proof procedure and detects the fraud, we assume the penalty for the malicious server is $S$, and the net gain for the validator is $R_C$, accounting for the difference between the reward and the cost of initiating the fraud proof procedure.

\begin{table}[htbp]
\centering
\label{not-collude-validator}
\begin{tabular}{|P{3.8cm}|P{3.2cm}|P{3.6cm}|}
\hline
Strategy & Server Fraud & Server Not Fraud \\ \hline
Validator Check & \( R_C,-S \) & \( -C,R-C \) \\ \hline
Validator Not Check & \( 0,R \) & \( 0,R-C \) \\ \hline
\end{tabular}
\end{table}

\noindent The table above gives a game in a bimatrix format, where the first number in each pair represents the utility to the validator, and the second number represents the utility to the server. We can calculate the probability for the undetected fraud is $(S+R-C)C/[(S+R)(R_C+C)]$ by the mixed strategy Nash Equilibrium, similar to the approach detailed in \cite{mamageishvili2023incentive}. Assuming $R_C=100C$, $R=1.2C$ and $S=150C$, the calculated probability of undetected fraud is $0.98\%$. This implies that if you request AI inference $50$ times per day, you can expect, on average, one undetected fraud approximately every $2$ days.

In contrast, in the spML scenario, assuming the fraction of the Byzantine nodes in the network $r=10\%$, by Proposition \ref{ai-prop}, the probability of triggering the challenge mechanism is $0.736\%$. This translates to only $0.736\%$ additional computational overhead in spML, enabling us to completely avoid fraud and eliminate the need for fraud proof procedures, if all nodes are rational. Hence, spML is the superior choice.

\section{Conclusions and Future Extensions}

In this study, we introduced the PoSP protocol and demonstrated a key application to decentralized AI inference network. Central to its design is the implementation of a unique Nash Equilibrium in pure strategies, ensuring that all rational participants within the network adhere to correct outputs, under certain assumptions. This protocol demonstrates superior performance when compared to zero-knowledge based protocols.

Looking ahead, further exploration into the application of the PoSP protocol within Layer 2 architectures holds promise, particularly by employing our method of sampling multiple nodes to recompute results. This approach can lead to a unique Nash Equilibrium in pure strategies, where every participant acts honestly, directly addressing and potentially solving the concerns highlighted in \cite{mamageishvili2023incentive}. Additionally, there is significant potential for applying PoSP as a verification mechanism within Actively Validated Services (AVS) in restaking protocols. Recent research highlights that fragmented validator sets can compromise security unless stakes are dynamically rebalanced \cite{nag2025economic}. Industry developments underscore a transition towards autonomously verifiable services, reinforcing the necessity for efficient, lightweight verification mechanisms. This exploration could lead to innovative applications that leverage the strengths of PoSP in ensuring robust, efficient, and reliable systems.

% \raluca{what do you mean by this? they are enhancing our security} \sq{We can use PoSP to verify the result output by AVS is correct}

{\bf Acknowledgements.}
	
We are thankful to Ari Juels, Zhongjing Wei, Zhe Ye, Jianzhu Yao, Chenghan Zhou, Yuchen Jin, and Dahlia Malkhi for helpful comments and conversations.

\newpage
\bibliographystyle{plain}
\bibliography{references.bib}

\newpage
\appendix

\section{Proof of Theorem \ref{main-theorem}}\label{appendix-proof}

First, we consider the expected payoff of the asserter if his result is correct. If the asserter's result is correct, all the validators, whether Byzantine or not, have a dominant strategy to output the correct result. Suppose $k$ is the number of nodes controlled by the asserter that are selected as validators when the challenge mechanism is triggered. If the asserter does not commit fraud, the expected payoff for the asserter in this round is at least 
$$(1-p)(R_A-C)+p\left(\frac{R_V}{n} \mathbb{E}\left[k\right] + R_A-C\right).$$

Then, we consider the expected payoff of the asserter if he commits fraud and outputs the incorrect result. If the challenge mechanism is not triggered, the asserter can get a payoff of $U_1$. However, if the challenge mechanism is triggered, fraud might go undetected and the asserter could earn $U_2$ only if all of the $n$ selected validators are Byzantine and submit the same result as the asserter. If $1\leq i<n$ out of $n$ selected validators are Byzantine or even collude with the asserter, their optimal strategy remains to act honestly and report the fraud in order to receive the validation reward. This is because the arbitration process will inevitably be triggered by the presence of at least one honest validator, making any dishonest action be penalized rather than rewarded.

Suppose $\rho \leq r$ is the fraction of Byzantine nodes that potentially submit the same result as the asserter out of all the nodes in the network; this includes nodes that may be controlled by the asserter as well. The expected payoff of the asserter is at most
$$(1-p)U_1+p\rho^nU_2+p\sum_{i=0}^{n-1} \binom{n}{i}\rho^i(1-\rho)^{n-i}\left(\frac{R_V}{n} \mathbb{E}\left[k \middle| m=i \right]-S\right),$$
where $m$ is the number of Byzantine nodes that potentially submit the same result as the asserter.

Hence, the system will have a unique Nash Equilibrium in pure strategies when the Byzantine asserter can obtain a greater profit if it does not commit fraud, i.e., if
$$(1-p)(R_A-C)+p\left(\frac{R_V}{n} \mathbb{E}\left[k\right] + R_A-C\right)>(1-p)U_1+p\rho^nU_2+p\sum_{i=0}^{n-1} \binom{n}{i}\rho^i(1-\rho)^{n-i}\left(\frac{R_V}{n} \mathbb{E}\left[k \middle| m=i \right]-S\right).$$
By rearranging this inequality, we can get
$$R_A+pS-(1-p)U_1-C>p\rho^n\left(U_2+S-\frac{R_V}{n}\mathbb{E}\left[k| m=n \right]\right).$$
Since we always have $\mathbb{E}\left[k| m=n\right] \geq 0$ and $\rho \leq r$, our system will have a unique Nash Equilibrium in pure strategies, if
$$R_A+pS-(1-p)U_1-C>pr^n\left(U_2+S\right),$$
which coincides Theorem \ref{main-theorem}.

\end{document}